\documentclass{llncs}
\usepackage{algorithm}
\usepackage{myalgorithmic}
\usepackage{epsfig}
\usepackage{latexsym}
\usepackage{amsmath}
\usepackage{amsfonts}
\usepackage{amssymb}
\usepackage{graphicx}
\usepackage{float}
\usepackage{subfig}
\usepackage{color}
\usepackage[mathscr]{euscript}

\newcommand{\mc}{\mathcal}

\newcommand{\C}{{\mathscr{C}}}
\newcommand{\Sc}{{\mathscr{S}}}

\newcommand{\Pc}{{\mathscr{P}}}
\newcommand{\F}{{\mc{F}}}
\newcommand{\X}{{\mc{X}}}

\newcommand{\ignore}[1]{}

\newcommand{\RETURNR}{\textbf{return}}

\renewcommand\algorithmicthen{}
\renewcommand\algorithmicdo{}

\newtheorem{obs}{Observation}

\begin{document}

\title{Shortest reconfiguration paths in the solution space of Boolean formulas}
\author
{
    Amer E. Mouawad\inst{1}\thanks{Research supported by the Natural Science and Engineering Research Council of Canada.} \and
    Naomi Nishimura\inst{1}$^{\star}$ \and
    Vinayak Pathak\inst{1} \and\\
    Venkatesh Raman\inst{2}
}
\institute
{
    David R. Cheriton School of Computer Science\\
    University of Waterloo, Ontario, Canada.\\
    \email{\{aabdomou, nishi, vpathak\}@uwaterloo.ca}
    \and
    The Institute of Mathematical Sciences\\
    Chennai, India.\\
    \email{vraman@imsc.res.in}
}
\maketitle

\begin{abstract}
Given a Boolean formula and a satisfying assignment, a
\emph{flip} is an operation that changes the value of a variable in
the assignment so that the resulting assignment remains satisfying.
We study the problem of computing the shortest sequence
of flips (if one exists) that transforms a given satisfying assignment $s$ to another
satisfying assignment $t$ of a Boolean formula. Earlier work characterized the
complexity of finding any (not necessarily the shortest) sequence of flips from one
satisfying assignment to another using Schaefer's framework for
classification of Boolean formulas. We build on it to
provide a trichotomy for the complexity of finding the
{\it shortest} sequence of flips and show that it
is either in P, NP-complete, or PSPACE-complete.

Our result adds to the small set of complexity results known
for {\it shortest reconfiguration sequence} problems by providing an example where the shortest sequence can be found
in polynomial time even though its length is not equal to the
symmetric difference of the values of the variables in $s$ and $t$.
This is in contrast to all reconfiguration
problems studied so far, where polynomial time algorithms for
computing the shortest path were known only for
cases where the path modified the symmetric difference only.
\end{abstract}

\section{Introduction}
\emph{Reconfiguration} problems study relationships between feasible
solutions to an instance of a computational problem and
have recently received significant attention~\cite{BC09,IDHPSUU11,KMM12,oursipec}.
The relationship between solutions is typically analyzed with respect to
a {\em reconfiguration step}, which specifies how one solution can be transformed into another.

For the problem of satisfiability, for example, one defines a reconfiguration
step to be a {\em flip} operation, that is, changing
the value of one variable in a satisfying assignment
such that the resulting assignment is also satisfying.
Most reconfiguration problems can be stated
concisely in terms of a graph---the
{\em reconfiguration graph}---that has a
node for each feasible solution and an undirected edge between two solutions if either
one can be formed from the other by a single reconfiguration
step. Thus for the reconfiguration of satisfiability~\cite{GKMP09},
there is a node for each satisfying assignment and an
edge whenever the {\em Hamming distance} between two assignments, i.e.
the number of variables in which the two assignments differ in value, is exactly one.

\section{Background and motivation}\label{sec:background}
\subsubsection{Reconfiguration}
In one of the earliest works on reconfiguration, Gopalan et al.~\cite{GKMP09} considered the
problem of deciding if a sequence of flips exists that can
reconfigure assignment $s$ to assignment $t$, both
satisfying a Boolean formula $\phi$; they showed that for any class of
formulas this question is either in P or is PSPACE-complete.
Since then, reconfiguration versions of various problems have been studied,
including maximum independent set, minimum vertex cover, maximum
matching, shortest path, graph colorability, and
many others~\cite{BC09,IDHPSUU11,IKD12,IKOZ12,oursipec}. Typical questions addressed in these works
include the structure or the complexity of determining
\begin{itemize}
\item
\emph{st-connectivity}: whether there is a path from $s$ to $t$
in the reconfiguration graph~\cite{BC09,IDHPSUU11,IKD12,IKOZ12} or
\item
\emph{connectivity}: whether the reconfiguration graph is connected~\cite{BB13temp,CVJ08,FHHH11} or
\item
upper bounds for the {\it diameter} of the reconfiguration graph~\cite{BJLPD11,BC09,IKD12}.
\end{itemize}

More recently, there has been interest in finding shortest paths (if one exists) as well
as in the parameterized complexity of reconfiguration problems~\cite{MNR14,oursipec}. Although some algorithms for
deciding st-connectivity also happen to compute the shortest path~\cite{IDHPSUU11}
(e.g. spanning trees, matchings), this is not the case for satisfiability of Boolean
formulas, the subject of this paper.
We study the question of computing the shortest flip sequence between
two satisfying assignments and complementing Gopalan et al.'s work, provide a partition of the set of
Boolean formulas into three equivalence classes where
the problem is in P, NP-complete, or PSPACE-complete.

Reconfiguration problems exhibit several recurring patterns. For example, most
reconfiguration versions of NP-complete decision problems are
PSPACE-complete~\cite{BC09,IDHPSUU11} (e.g. maximum independent set) whereas versions of
problems in P are in P~\cite{IDHPSUU11} (e.g. maximum matching). Known exceptions
include the shortest path and 3-coloring problems; the former is in P but has a
reconfiguration version that is PSPACE-complete~\cite{Bonsma12} and the latter is NP-complete
but has a reconfiguration version that is in P~\cite{CVJ11}.
Another recurring pattern is a connection between the st-connectivity problem (in P or
PSPACE-complete) and the diameter of the reconfiguration graph
(polynomial or exponential, respectively).

Most relevant to our work is the pattern that the only polynomial-time
algorithms known for finding the shortest reconfiguration path have the property that they make no
changes to parts of the solution common to $s$ and $t$. For
trees and cactus graphs, the shortest path between maximum independent
sets $s$ and $t$ never removes vertices in $s \cap t$~\cite{MNR14}. In the sequence
of flips for 2CNF formulas (the only class for which a polynomial-time
algorithm for shortest reconfiguration path of satisfiability was previously
known), the only variables flipped are those whose values are different in
$s$ and $t$~\cite{GKMP09}.
To the best of our knowledge, our results on computing the shortest path in a reconfiguration graph
for satisfiability provide the first exception to this pattern. In
particular, we provide a class of Boolean formulas where the shortest
reconfiguration path can flip variables that have the same values in $s$ and $t$ and yet the path can be
computed in polynomial time.  Insights from our results may lead to a better
understanding of the role of the symmetric difference in computing shortest reconfiguration paths.

\subsubsection{Flips in triangulations}
The problem of computing the shortest reconfiguration sequence has a long
history in the field of triangulations~\cite{Bern-Eppstein,Bose-Verdonschot,DRS10,Lawson-72},
although it has not been studied with this name. The reconfiguration of triangulations of a
convex polygon makes use of a \emph{flip} operation that replaces one
diagonal with another. It is known that one can always transform one
triangulation of a polygon to another~\cite{Lawson-72}; therefore, research has focused on
the complexity of finding shortest reconfiguration paths, where
results have been obtained for planar point
sets, simple polygons, convex polygons and triangulations where edges have
labels~\cite{AMP12,BLPV13,CulikWood,Pilz}. This problem is identical to reconfiguring
independent sets for a certain kind of a graph, providing an
example where although st-connectivity and connectivity are both
trivially solvable, for many cases the complexity of shortest
reconfiguration path has been open for more than 40 years~\cite{CulikWood}.

Interestingly, one distinction between the case of convex polygons, which is open,
and the case of simple polygons, which is NP-complete, is that the
former but not the latter has the property that the shortest flip
sequence never flips a diagonal shared by $s$ and $t$. This adds to the motivation
for studying reconfiguration problems where the shortest reconfiguration path
can be found in polynomial time even though the path flips objects that are already common between $s$ and $t$.

\subsubsection{Reconfiguration on Boolean formulas and Schaefer's framework}
Schaefer's~\cite{Schaefer} framework provides a way to classify Boolean formulas and was first used by Schaefer to show that for any class that can be defined using the framework, deciding whether a formula of that class has a satisfying assignment is either in P or NP-complete.

Schaefer's framework has previously been used by Gopalan et al.~\cite{GKMP09} and
Schwerdtfeger~\cite{Schwerdt} in the context of reconfiguration, where they provide a similar characterization for st-connectivity and connectivity of the reconfiguration graph, respectively. In our work, we provide a similar complete characterization for finding the shortest reconfiguration path in terms of classes for which it is in P, NP-complete, or PSPACE-complete.
In particular, our results imply that
there are classes where we can compute shortest reconfiguration paths even when the path
flips variables that have the same value in both $s$ and $t$.

\subsubsection{Shortest paths in large graphs}
A labelled hypercube in $n$ dimensions exhibits a shortest path finding algorithm that takes time logarithmic in the size of the graph---simply compute the Hamming distance between the two vertices. Partial cubes are subgraphs of the hypercube where the same property holds~\cite{DBLP:books/daglib/0019809}. In general, a \emph{distance labeling scheme}~\cite{GPRR01,P00,WP11} is an assignment of bit vectors to the vertices of a given graph such that the length of the shortest path between two vertices can be computed just from the bit vectors assigned to the two vertices. Small distance labels provide efficient shortest path algorithms for large graphs. 

Interestingly, the
reconfiguration graph of satisfying assignments of 2CNF formulas is
known to be a partial cube. One consequence of our results is the identification of a new
class of subgraphs of the hypercube (reconfiguration graphs of navigable formulas, as defined in Section~\ref{sec-prelim}) where shortest paths can be found efficiently. Our class is fundamentally more complex than partial cubes in the sense that the distance
between two vertices is not merely the Hamming distance between their labels.

\section{Computing shortest reconfiguration paths}
\subsection{Preliminaries}\label{sec-prelim}
We use terminology originally introduced by Schaefer~\cite{Schaefer}
and adapted to reconfiguration by Gopalan et al.~\cite{GKMP09} and Schwerdtfeger~\cite{Schwerdt}.

A \emph{$k$-ary Boolean logical relation} (or \emph{relation} for short)
$R$ is defined as a subset of $\{0,1\}^k$, where $k \geq 1$. Each $i\in \{1,\ldots,k\}$
can be interpreted as a variable of $R$ such that $R$ specifies exactly which
assignments of values to the variables are to be considered satisfying.

For any $k$-ary relation $R$ and positive integer $k' \leq k$, we define
a $k'$-ary \emph{restriction} of $R$ to be any $k'$-ary relation $R'$ that
can be obtained from $R$ by substitution with constants and identification of variables.
More precisely, let $X:\{1,\ldots,k\}\rightarrow\{1,\ldots , k'\}\cup\{c_0, c_1\}$ be a mapping from
the variables of $R$ to the variables of $R'$ and the constants 0 and 1.
Any such $X$ defines a mapping $f_X:\{0,1\}^{k'} \rightarrow\{0,1\}^{k}$ as follows.
For $r\in\{0,1\}^{k'}$, let $f_X(r)$ be the $k$-bit vector whose $i^{th}$ bit is 0
if $X(i) = c_0$, 1 if $X(i) = c_1$ and equal to the $X(i)^{th}$ bit of $r$ otherwise.
We say that a $k'$-ary relation $R'$ is a restriction of $R$ with respect
to $X:\{1,\ldots,k\}\rightarrow\{1,\ldots,k'\}\cup\{c_0,c_1\}$ if $r \in R' \Leftrightarrow f_X(r) \in R$.

A Boolean formula $\phi$ over a set $\{x_1,\ldots ,x_n\}$ of variables defines a
relation $R_\phi$ as follows. For any $n$-bit vector $v\in\{0,1\}^n$, we
interpret $v$ as the assignment to the variables of $\phi$ where $x_i$ is
set to be equal to the $i^{th}$ bit of $v$. We then say that $v\in R_\phi$ if and only if $v$ is a satisfying assignment.

A \emph{CNF formula} is a Boolean formula of the form $C_1 \wedge \ldots \wedge C_m$,
where each $C_i$, $1 \leq i \leq m$, is a \emph{clause} consisting of a finite disjunction of
\emph{literals} (variables or negated variables). A \emph{$k$CNF formula}, $k \geq 1$, is a CNF formula
where each clause has at most $k$ literals. A CNF formula is \emph{Horn} (\emph{dual Horn})
if each clause has at most one positive (negative) literal.

For a finite set of relations $\mc{S}$, a \emph{CNF($\mc{S}$) formula} over a set of $n$ variables
$\{x_1,\ldots,x_n\}$ is a finite collection $\{C_1,\ldots,C_m\}$ of clauses.
Each $C_i$, $1 \leq i \leq m$, is defined by a tuple $(R_i, X_i)$, where $R_i$ is
a $k_i$-ary relation in $\mc{S}$ and $X_i:\{1,\ldots,k_i\}\rightarrow\{1,\ldots,n\}\cup\{c_0,c_1\}$ is a function.
Each $X_i$ defines a mapping $f_{X_i}:\{0,1\}^n\rightarrow\{0, 1\}^{k_i}$ and we say that
an assignment $v$ to the variables satisfies $\phi$ if and only if
for all $i\in \{1,\ldots,m\}$, $f_{X_i}(v)\in R_i$. For any variable $x_j$, we say
that $x_j$ \emph{appears in} clause $C_i$ if $X_i(q) = j$ for some $q \in \{1,\ldots,k_i\}$ and for any
assignment $v$ to the variables of $\phi$, we say that $f_{X_i}(v)$ is the assignment induced by $v$ on $R_i$.

For example, to represent the class 3CNF in Schaefer's framework, we specify $\mc{S}$ as
follows. Let $R^0 = \{0,1\}^3\backslash\{000\}$, $R^1 =
\{0,1\}^3\backslash\{100\}$, $R^2 = \{0,1\}^3\backslash\{110\}$,
$R^3 = \{0,1\}^3\backslash\{111\}$, and $\mc{S} = \{R^0, R^1,
R^2, R^3\}$. Since $R^i$ can be used to represent all 3-clauses with
exactly $i$ negative literals (regardless of the positions in which they appear in a clause),
clearly CNF($\mc{S}$) is exactly the class of 3CNF formulas.

Below we define some classes of relations used in the literature and
relevant to our work. Note that componentwise bijunctive, OR-free and NAND-free
were first defined by Gopalan et al.~\cite{GKMP09}. Schwerdtfeger~\cite{Schwerdt}
later modified them slightly and defined safely component-wise bijunctive, safely
OR-free and safely NAND-free. We reuse the names componentwise bijunctive, OR-free
and NAND-free for Schwerdtfeger's safely component-wise bijunctive, safely OR-free and safely NAND-free respectively.

\begin{definition}\label{def-relations}
For a $k$-ary relation $R$:
\begin{itemize}
\item $R$ is \emph{bijunctive} if it is the set of satisfying assignments of a 2CNF formula.
\item $R$ is \emph{Horn} (\emph{dual Horn}) if it is the set of satisfying assignments of a Horn (dual Horn) formula.
\item $R$ is \emph{affine} if it is the set of satisfying assignments of a formula $x_{i_1} \oplus \ldots \oplus x_{i_h} \oplus c$,
with $i_1,\ldots,i_h \in \{1,\ldots,k\}$ and $c \in \{0,1\}$. Here $\oplus$ denote the {\it exclusive OR} operation
which evaluates to $1$ when {\it exactly} one of the values it operates on is $1$ and evaluates to $0$ otherwise.
\item $R$ is \emph{componentwise bijunctive} if every connected component of the
reconfiguration graph of $R$ and of the reconfiguration graph of every restriction $R'$ of $R$ induces a bijunctive relation.
\item $R$ is \emph{OR-free} (\emph{NAND-free}) if there does not exist a restriction $R'$ of $R$ such that $R' = \{01, 10, 11\}$
($R' = \{01, 10, 00\}$).
\end{itemize}
\end{definition}

Using his framework, Schaefer showed that SAT($\mc{S}$)---the problem of deciding if a
CNF($\mc{S}$) formula has a satisfying assignment---is in P if every relation in $\mc{S}$ is bijunctive,
Horn, dual Horn, or affine, and is NP-complete otherwise.
The result is remarkable because it divides a large set of problems
into two equivalence classes based on their computational
complexity, which is the opposite of what one might
expect due to Ladner's theorem~\cite{Arora:2009:CCM:1540612}.

Since Schaefer's original paper, a myriad of problems about Boolean
formulas have been analyzed, and similar divisions into equivalence
classes obtained~\cite{creignou2001complexity}. Gopalan et al.'s work~\cite{GKMP09}, with corrections presented
by Schwerdtfeger~\cite{Schwerdt}, shows a dichotomy for the problem of deciding
whether a reconfiguration path exists between two satisfying assignments
of a CNF($\mc{S}$) formula.

They call a set $\mc{S}$ of relations \emph{tight} if
\begin{itemize}
\item
all relations in $\mc{S}$ are componentwise bijunctive, or
\item
all relations in $\mc{S}$ are OR-free, or
\item
all relations in $\mc{S}$ are NAND-free.
\end{itemize}
They showed that the st-connectivity
problem on CNF($\mc{S}$) formulas is in P if $\mc{S}$ is tight and PSPACE-complete otherwise.

Our trichotomy relies on a new class of formulas that subdivides the tight classes
into those for which computing the shortest reconfiguration path
can be done in polynomial time and those for which it is NP-complete.

\begin{definition}
For a $k$-ary relation $R$:
\begin{itemize}
\item $R$ is \emph{Horn-free} if there does not exist a restriction $R'$ of $R$
such that $R' = \{0,1\}^3 \setminus \{011\}$,
or equivalently, $R'$ is the set of all satisfying assignments
of the clause $(x \vee \overline{y} \vee \overline{z})$ for some three variables $x$, $y$, and $z$.
\item $R$ is \emph{dual-Horn-free} if there does not exist a restriction $R'$ of $R$
such that $R' = \{0,1\}^3 \setminus \{100\}$,
or equivalently, $R'$ is the set of all satisfying assignments of
the clause $(\overline{x} \vee y \vee z)$ for some three variables $x$, $y$, and $z$.
\end{itemize}
\end{definition}

The following is a useful observation.

\begin{obs}
For $k \geq 3$ and $R$ a $k$-ary relation, if $R$ is OR-free then it is dual-Horn-free.
Similarly, if $R$ is NAND-free then it is Horn-free.
\end{obs}
\begin{proof}
Assume that $R$ is OR-free but not dual-Horn-free.
Then there exists a restriction $R'$ of $R$
such that $R' = \{0,1\}^3 \setminus \{100\}$. It is easy to see that,
from $R'$, one can obtain $R'' = \{01, 10, 11\}$ by setting one of
the three variables in $R'$ to $0$, resulting in a contradiction.
A similar proof shows that NAND-free relations are Horn-free.
\qed
\end{proof}

\begin{definition}
\label{def-tree-like}
We call a set $\mc{S}$ of relations \emph{navigable} if one of the following holds:
\begin{itemize}
\item[(1)] All relations in $\mc{S}$ are OR-free and Horn-free.
\item[(2)] All relations in $\mc{S}$ are NAND-free and dual-Horn-free.
\item[(3)] All relations in $\mc{S}$ are component-wise bijunctive.
\end{itemize}
\end{definition}

It is clear that if $\mc{S}$ is navigable, then it is also tight. Our main result is the following trichotomy.
\begin{theorem}
\label{thm:main}
For a CNF($\mathcal{S}$) formula $\phi$ and two
satisfying assignments $s$ and $t$, the problem of
computing the shortest reconfiguration path between
$s$ and $t$ is in P if $\mc{S}$ is navigable, NP-complete if $\mc{S}$ is
tight but not navigable and PSPACE-complete otherwise.
\end{theorem}

In the next section, we establish the hardness results; the rest of the paper
is devoted to develop our polynomial time algorithm for navigable formulas.
Interestingly, unlike previous classification results, while the NP-completeness
result in our case turns out to be easier, the polynomial time algorithm is quite involved.

\subsection{The hard cases}
Gopalan et al.~\cite{GKMP09} showed that if $\mc{S}$ is not tight, then
st-connectivity is PSPACE-complete for CNF($\mc{S}$). This implies that
finding the shortest reconfiguration path is also PSPACE-complete for such classes of formulas.

\begin{theorem}
If $\mc{S}$ is tight but not navigable, then finding the shortest reconfiguration
path on CNF($\mc{S}$) formulas is NP-complete.
\end{theorem}
\begin{proof}
The problem is in NP because the diameter of the reconfiguration graph is
polynomial for all tight formulas, as shown by Gopalan et al.~\cite{GKMP09}.
We now prove that it is, in fact, NP-complete.

As $\mc{S}$ is tight but not navigable, all relations in $\mc{S}$ are OR-free or all relations in $\mc{S}$ are NAND-free.
Let us assume that all relations in $\mc{S}$ are NAND-free (we handle the other case later).
Then, as $\mc{S}$ is not navigable, there exists a relation which is dual-Horn.

We show a reduction from \textsc{Vertex Cover} to such a CNF($\mc{S}$) formula.
Given an instance $(G = (V, E), k)$
of \textsc{Vertex Cover}, we create a variable $x_v$ for each $v \in V$. For
each edge $e=(u,v) \in E$, we create two new variables $y_e$ and $z_e$ and the
clauses $(y_e \vee \overline{z_e} \vee x_{u})$ and $(z_e \vee \overline{y_e} \vee
x_{v})$. The resulting formula $F(G)$ has $|V|+2|E|$ variables and $2|E|$
clauses.

It is easy to see that all the relations of $F(G)$ are NAND-free (as we cannot set
the values of all but two of their variables to get a NAND relation), however none of them is
dual-Horn-free (as each clause has two positive literals).
Hence the formula $F(G)$ is tight but not navigable.

Let $s$ be the satisfying assignment for the formula with all variables set to $0$, and let $t$
be the satisfying assignment with all the variables $x_v , v \in V$ set to $0$ and the rest set to $1$.
If $G$ has a vertex cover $S$ of size at most $k$, then we can form a
reconfiguration sequence of length at most $2|E|+2k$ from $s$ to $t$
by flipping each $x_v, v \in S$ from $0$ to $1$, flipping  the $y_e$
and $z_e$ variables, and then flipping each $x_v, v \in S$ back from $1$ to
$0$.  To show that such a reconfiguration sequence exists only if there exists such a
vertex cover, we observe that if neither $x_u$ nor $x_v$ has
been flipped to 1, neither $y_e$ nor $z_e$ can be flipped to 1 while keeping the formula satisfied at the intermediate steps.

To show hardness when all relations in $\mc{S}$ are OR-free but not Horn-free, we give a reduction
from Independent set. Given $G=(V,E)$ and an integer $k$, we create, as before, a
variable $x_v$ for each $v \in V$ and two variables $y_e$ and $z_e$ for each $e \in E$.
For each edge $e=(u,v) \in E$, we create the clauses
$(y_e \vee \overline{z_e} \vee \overline{x_{u}})$ and $(\overline{y_e} \vee z_e \vee \overline{x_v})$.
Clearly, all the relations of the formula are OR-free, and none of them is Horn-free.

We let $s$ be the satisfying assignment that sets all the variables to $1$, and $t$ be
the satisfying assignment that sets all the variables to $0$ except the variables $x_v v \in V$ that are set to $1$.
If $G$ has an independent set of size at least $k$, then it has a vertex cover $S$ of size at most $n-k$, then we can form a
reconfiguration sequence of length at most $2|E|+2(n-k)$ from $s$ to $t$
by flipping each $x_v, v \in S$ from $1$ to $0$, flipping  the $y_e$
and $z_e$ variables, and then flipping each $x_v, v \in S$ back from $0$ to
$1$.  To show that such a reconfiguration sequence exists only if there exists such a
vertex cover (of size $n-k$), we observe that if neither $x_u$ nor $x_v$ has
been flipped to 0, neither $y_e$ nor $z_e$ can be flipped to 0 while keeping the formula satisfied at the intermediate steps.
\qed
\end{proof}

\subsection{The polynomial-time algorithm for navigable formulas}
In this section, we give the polynomial time algorithm to find the shortest
reconfiguration sequence between two satisfying assignments of a navigable formula.

Gopalan et al. gave a polynomial-time algorithm for finding the shortest
reconfiguration path in component-wise bijunctive formulas. The path, in this case, flips
only variables that have different values in $s$ and $t$. The NP-completeness proof
from the previous section crucially relies on the fact that we need to
flip variables with common values; in fact, the hardness lies in deciding precisely which
common variables need to be flipped. Thus it is tempting to conjecture
that hardness for shortest reconfiguration path is caused by
relations where the shortest distance is not always equal to the Hamming distance.

Interestingly, this intuition is wrong. The reconfiguration graph for
the relation $R=\{000, 001, 101, 111, 110\}$ is a path of length four,
where for 000 and 110 the shortest path is of length four but the
Hamming distance is two. However, we can
find shortest reconfiguration paths in formulas built out of $R$ in
polynomial time, the exact reason for which will become clear in our general description of the algorithm. The
intuitive reason is that there are very few candidates for shortest
paths; if we restrict our attention to a single clause built out of $R$, then
there exists a unique path to follow. It then suffices to determine whether there exist
two clauses for which the prescribed paths are in conflict. In general, our proof
relies on showing that even if there does not exist a unique path, the set of all possible
paths between two satisfying assignments of a navigable formula is not diverse enough to
make the problem computationally hard. We show that the set of all possible
paths can be characterized using a partial order on the set of flips.

\subsubsection{Notation}
Our results make use of two different views of the problem
(graph theoretic and algebraic), and hence two sets of notation.

The graph-theoretic view consists of the reconfiguration graph $G_R$
that has a node for each Boolean string $s\in R$ and an edge whenever the
Hamming distance between the two strings is exactly one. We call a
path from $s$ to $t$ \emph{monotonically increasing} if the Hamming
weights of the vertices on the path increase monotonically as we
go from $s$ to $t$, and define a \emph{monotonically decreasing}
path similarly. A path is \emph{canonical} if it consists of a
monotonically increasing path followed by a monotonically decreasing path.

The algebraic view consists of a \emph{token system}~\cite{DBLP:books/daglib/0019809}
consisting of a set $\mathscr{S}$ of states and a set $\tau$ of
tokens. The tokens specify the rules of transition between states.
Each token $t\in\tau$ is a function that maps $\mathscr{S}$ to
itself. Given a $k$-ary relation $R$, we define a token system as
follows. The set $\mathscr{S}$ of states consists of all the elements
of $R$ and a special state $s^*$ called the \emph{invalid state} that captures all the unsatisfying assignments of the formula. The
set $\tau$ of tokens is the set $\{x_1^+, \ldots , x_k^+\}\cup
\{x_1^-, \ldots , x_k^-\}$, where $x_i^+$ denotes a flip of variable
$x_i$ from 0 to 1, which we call a \emph{positive flip}, and denote the sign of the flip as positive,
and $x_i^-$ denotes a flip of
variable $x_i$ from 1 to 0, which we call a \emph{negative flip} and denote the sign of the flip as negative.

To complete the description of the token system, we need to specify
the function to which each token corresponds. For $x_i^+\in\tau$ and
$s\in\mathscr{S}$, $x_i^+(s^*) = s^*$, $x_i^+(s) = s^*$ if the value of
variable $x_i$ in $s$ is 1, $x_i^+(s) = s'$ if the value of variable
$x_i$ in $s$ is 0 and the bit string $s'$ obtained on flipping it to 1
lies in $R$, and $x_i^+(s) = s^*$ if the value of
variable $x_i$ in $s$ is 0 and the bit string $s'$ obtained on
flipping it to 1 does not lie in $R$.  The function $x_i^-$ is
defined analogously.  In the rest of this article, we will use the
word ``flip" instead of ``token", and we will
use the words ``state,'' ``vertex,'' and ``satisfying assignment'' interchangeably.

A sequence of flips also defines a function, that is, the composition
of all the functions in the sequence. We call a flip sequence
\emph{invalid} at a given state $s$ if the sequence applied to $s$
results in invalid state $s^*$, and \emph{valid} otherwise. Two flip
sequences are \emph{equivalent} if they result in the same final state
when applied to the same starting state.  Finally, we call a flip
sequence \emph{canonical} if all positive flips in it occur before all
the negative flips. That is, the path from its first state (node) to the last is a canonical path.
Note that in any canonical flip sequence, each
flip occurs at most once. Given two states $s, t\in\Sc$, we say that a set $\C$
of flips \emph{transforms $s$ to $t$} if the elements of $\C$ can be arranged
in some order such that the resulting flip sequence transforms $s$ to
$t$. For a given state $s$ and flip set $\C$, we say $\C$ is \emph{valid}
if the elements of $\C$ can be arranged
in some order such that the resulting flip sequence applied to $s$
results in a valid state.

We describe a flip sequence simply by listing the flips in order. The flip
sequence formed by removing flip $f$ from $\F$ is
denoted $\F \setminus f$. The flip sequence obtained by reversing $\F$ is $\F^{-1}$,
and by performing $\F_1$ followed by $\F_2$ is
$\F_1\cdot\F_2$. We use $\C(\F)$ to denote the set of flips that
appear in $\F$. A flip sequence (set) consisting of only positive flips
will be called a \emph{positive flip sequence (set)}. We use $\F_0$ to denote an
empty flip sequence and, by convention, define it to be valid. For a flip
sequence $\F$, if $f\in\F$ appears before $f'\in\F$ in the sequence, then we say $f<_\F f'$.
For a tuple $t=(x_{i_1},\ldots,x_{i_d})$ of variables and a state $s$, we use $s^t$
to denote the string of values restricted to $x_{i_1}, \ldots, x_{i_d}$.
\subsubsection{Overview of the algorithm}
For a CNF($\mc{S}$) formula $\phi$ and two satisfying assignments $s$ and $t$,
if every relation in $\mc{S}$ is componentwise
bijunctive, then the algorithm of Gopalan et al. gives a polynomial time
algorithm to find a shortest path between $s$ and $t$.
Hence we will assume that every relation in $\phi$ is NAND-free and dual-Horn-free.

There are two crucial properties of NAND-free and dual-Horn-free relations
that help us design a polynomial time algorithm. First, we show in
Lemma~\ref{lemma-canonical} (originally proved by Gopalan et al.) that in a
NAND-free relation, any valid flip sequence from $s$ to $t$ can be transformed into
an equivalent canonical flip sequence, where all positive flips are performed before
all negative flips. Since the vertex reached after performing all
the positive flips has a larger Hamming weight than both $s$ and $t$, it can be viewed as
a common ancestor, and thus the shortest reconfiguration sequence defines a ``least common ancestor''.
Note however that finding such a least common ancestor may not be easy, as
not all orderings of those positive flips may be valid.

Next, we show that if the relation is both NAND-free and dual-Horn-free, then the set
of positive valid flip sets starting from a given satisfying assignment $s$ forms a distributive lattice~\cite{birkhoff1937}.
Thus using Birkhoff's representation theorem~\cite{birkhoff1937}, we obtain a partial order among the positive
flips that any valid flip sequence must follow. Moreover, since the positive
valid flip sets have a lattice structure, $s$ and $t$ have a
unique least common ancestor. We use the partial order to find it.

If every relation in $\mc{S}$ is OR-free and Horn-free, similar properties hold but the
role of positive and negative flips is ``reversed''.
In other words, in an OR-free relation, any valid flip sequence from $s$ to $t$ can be transformed into
an equivalent flip sequence, where all negative flips are performed before
all positive flips. Moreover, if the relation is both OR-free and Horn-free, the set of negative
flips becomes characterizable by a partial order.
Hence, we will only consider properties of NAND-free and dual-Horn-free relations.
Our algorithm for NAND-free and dual-Horn-free relations can easily be modified to handle
OR-free and Horn-free relations.

\subsubsection{The token system of NAND-free relations}
We begin by proving some useful properties of the token system formed by NAND-free
relations.

\begin{lemma}\label{lemma-swap}
For $R$ a NAND-free relation and $\F = f_1\ldots f_q$ a valid
flip sequence at $s \in R$,
if there exists $i\in \{1,\ldots,q-1\}$ such
that $f_i = x^-$ is a negative flip and $f_{i+1} = y^+$ is a positive flip,
with $x \neq y$,
then the sequence $\F' = f_1 \ldots f_{i-1}f_{i+1}f_i\ldots f_q$
is also valid at $s$ and is equivalent to $\F$, i.e., swapping $f_{i}$
and $f_{i+1}$ results in an equivalent flip sequence.
\end{lemma}

\begin{proof}
Let $u$ be the state right before applying $f_i$ in $\F$, $v=f_i(u)$ be the state after
applying $f_i$ but before applying $f_{i+1}$, and $w=f_{i+1}(v)$ be the one after applying
$f_{i+1}$. Thus it is clear that $u^{(x, y)} = 10, v^{(x, y)} = 00$, and
$w^{(x, y)} = 01$. Also, notice that since no other variables are flipped
between $u, v$, and $w$, the values of all variables other than $x$ and $y$ remain the same in the states $u$, $v$ and $w$.
Let $t$ be the Boolean string whose value is the same
as $u, v$, and $w$ on all variables except $x$ and $y$ and
$t^{(x, y)} = 11$. If $t\notin R$, then the substitution
described above gives us the relation $\{10, 00, 01\}$ on $x$
and $y$, which is precisely the NAND relation.
Since $R$ is NAND-free, $t\in R$ (Figure~\ref{figure-lemmas} (a))
and thus we can replace the path $u\rightarrow v \rightarrow w$ with the
path $u\rightarrow t\rightarrow w$. This is equivalent to swapping the flips $f_{i+1}$ and $f_i$.
\qed
\end{proof}

Lemma~\ref{lemma-canonical} now follows immediately. It
shows (first proved by Gopalan et al.~\cite{GKMP09}) that any valid flip sequence can be made canonical.

\begin{lemma}\label{lemma-canonical}
For $R$ a NAND-free relation, if $\F$ is a valid sequence at $s \in R$,
then there exists a valid canonical sequence $\F'$ equivalent
to $\F$ such that $\C(\F')\subseteq \C(\F)$ and, for any
two flips $f_1, f_2\in\F'$ of the same sign, if $f_1<_{\F'} f_2$
then $f_1<_{\F} f_2$, i.e., the relative order among flips of the same sign is preserved.
\end{lemma}

\begin{proof}
If $\F$ is not canonical, it must have a negative flip followed by a positive flip somewhere. If both flips act on the same variable, we cancel them out; otherwise, we swap them using the proof of Lemma~\ref{lemma-swap}. Doing this repeatedly gives us the required canonical sequence $\F'$. The order among the flips of the same sign is preserved since we never swap two flips of the same sign.\qed
\end{proof}

\begin{lemma}
\label{lemma-union}
For $R$ a NAND-free relation,
if $\C_1$ and $\C_2$ are two positive flip sets
that are valid at $s \in R$, then $\C_1\cup\C_2$ is also a valid flip set at $s$.
\end{lemma}

\begin{proof}
Let $u = \F_1(s)$ and $v = \F_2(s)$, where $\F_1$ and $\F_2$ are valid
flip sequences such that $\C(F_1) = \C_1$ and $\C(F_2) = \C_2$.
Clearly, $\F_1^{-1}\cdot\F_2$ is a valid flip
sequence from $u$ to $v$. Thus, we can apply Lemma~\ref{lemma-canonical} to the
sequence $\F_1^{-1}\cdot\F_2$ to transform it into the canonical sequence $\F$.
Let $\F^+$ denote the prefix of $\F$ that contains all the positive flips. It is clear
that $\F_1\cdot\F^+$ is a valid flip sequence at $s$ and $\C(\F_1\cdot\F^+) = \C_1\cup\C_2$.
\qed
\end{proof}

Later, we prove a similar lemma for the intersection of two flip sets, but for dual-Horn-free relations.
We conclude this subsection with a lemma that shows that if two disjoint
flips sets are valid at a state, we can, in some sense, perform them (the two sets of flips) one after the other in either order.

\begin{lemma}
\label{lemma-another-nand-free}
For $R$ a NAND-free relation and $\F_1$ and $\F_2$ two positive flip sequences
that are valid at $s\in R$, if $\C(\F_1)\cap\C(\F_2)=\emptyset$, then $\F_1$ is
valid at $\F_2(s)$ and $\F_2$ is valid at $\F_1(s)$.
\end{lemma}

\begin{proof}
Consider the sequence $\F_2^{-1}\cdot\F_1$ that
transforms $\F_2(s)$ to $\F_1(s)$. Applying Lemma~\ref{lemma-canonical} to it, we
obtain the canonical flip sequence $\F_1\cdot\F_2^{-1}$.
Thus $\F_1$ is valid at $\F_2(s)$. Using the same argument on the
sequence $\F_1^{-1}\cdot\F_2$ proves the other claim.
\qed
\end{proof}

\subsubsection{The token system of dual-Horn-free relations}
In this section, we establish stronger properties with the assumption that
$R$ is not only NAND-free, but is also dual-Horn-free. We begin by establishing a simple property of relations that are NAND-free and dual-Horn-free in the following lemma.
\ignore{
The first lemma states that if we have two length-two positive flip
sequences starting from some state, both ending at the same edge label, then the
two flips on each of the two paths can be swapped.
}

\begin{lemma}\label{lemma-tree-like}
Let $R$ be a NAND-free and dual-Horn-free relation and
$s, t_1, t_2 \in R$ be three distinct states such that the flip sequence $\F_1 = x_k^+x_i^+$ transforms
$s$ to $t_1$, the flip sequence $\F_2 = x_j^+x_i^+$ transforms $s$ to $t_2$, and $x_k \neq x_j$.
Then the sequence $\F_1' = x_i^+x_k^+$ also transforms $s$ to $t_1$ and the
sequence $\F_2' = x_i^+x_j^+$ also transforms $s$ to $t_2$, i.e., we can
swap the flips in both $\F_1$ and $\F_2$.
\end{lemma}

\begin{proof}
For $u_1=x_k^+(s)$ and $u_2=x_j^+(s)$, the sequence $x_j^-x_k^+$
transforms $u_2$ to $u_1$. We can reorder the sequence to obtain
$x_k^+x_j^-$, using Lemma~\ref{lemma-swap}. For $v=x_k^+(u_2)$, we can use a
similar argument to show that $x_i^+$ is a valid flip at $v$; we let
$w=x_i^+(v)$. The values of variables $x_i$, $x_j$, and $x_k$ at states
$s, u_1, u_2, t_1, t_2, v$, and $w$ form exactly the seven satisfying
assignments $\{000, 001, 010, 101, 110, 011, 111\}$ of the
dual-Horn clause $(\overline{x_i}\vee x_j\vee x_k)$ (Figure~\ref{figure-lemmas} (b)).
But since $R$ is dual-Horn-free, there must also exist the state $v'$
for which $x_i = 1, x_j=0, x_k=0$. The path $s\rightarrow v'\rightarrow t_1$
gives the sequence $x_i^+x_k^+$ and the path $s\rightarrow v'\rightarrow t_2$
gives the sequence $x_i^+x_j^+$.
\qed
\end{proof}

\begin{figure}
\hspace{15 mm}
\subfloat[]{\includegraphics[scale=0.35]{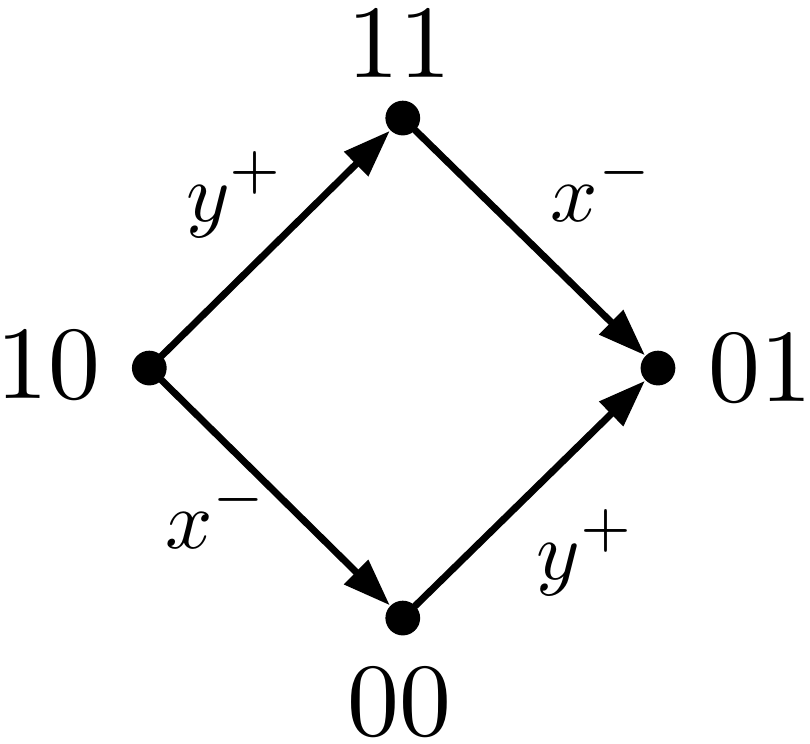}}
\hspace{25 mm}
\subfloat[]{\includegraphics[scale=0.35]{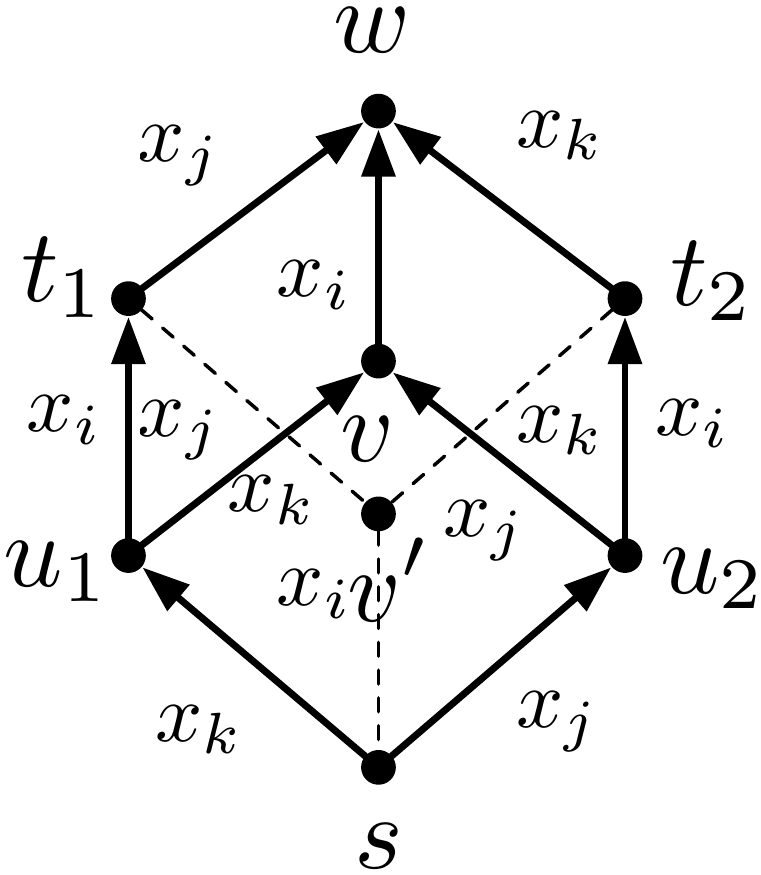}}
\caption{(a) Example for Lemma~\ref{lemma-swap} (b) Example for Lemma~\ref{lemma-tree-like}}
\label{figure-lemmas}
\end{figure}

The seemingly innocuous lemma above turns out to be very powerful. In the following sequence of lemmas, we cleverly build on top of it to eventually prove that the set of all positive valid flip sets starting from an assignment $s$ forms a distributive lattice. The lattice structure then helps us formulate a polynomial time algorithm for computing the shortest reconfiguration path.

\begin{lemma}\label{lemma-another-tree-like}
Let $R$ be a NAND-free and dual-Horn-free relation and $s, t\in R$ be two satisfying assignments such that $x^+y^+$ is a valid flip sequence at $s$ and $y^+$ is a valid flip at $t$. Furthermore, let $\F$ be a positive flip sequence such that $\F(s) = t$ and $x^+\not\in\C(\F)$. Then, the sequence $y^+x^+$ must also be valid at $s$.
\end{lemma}

\begin{proof}
Let $v$ be the vertex with smallest Hamming weight on the path corresponding to $\F$
from $s$ to $t$ (including $s$ and $t$)
at which $y^+$ is a valid flip. Let $\F_1 = x^+y^+$ and let $\F_2$ be the
positive flip sequence that transforms $s$ to $v$, i.e. $v = \F_2(s)$.  Note that
$\C(\F_1)\cap\C(\F_2)=\emptyset$, as neither $x^+$ nor $y^+$ can appear in $\C(\F_2)$ (See Figure~\ref{figure-lemmas-seven-eight}(a)).
If $v = s$, we are done; then let us assume this not to be the case.
Let $u$ be the vertex immediately
before $v$ on the path from $s$ to $t$ and let $z^+(u) = v$.
Since $\C(\F_1)\cap\C(\F_2)=\emptyset$ and $\C(\F_1) \cap \{\C(\F_2)\setminus\{z^+\}\}=\emptyset$,
we can apply Lemma~\ref{lemma-another-nand-free} at $s$, which implies that
$x^+y^+$ must be valid at both $u$ and $v$. Now we use Lemma~\ref{lemma-tree-like} at $u$. Since both
$x^+y^+$ and $z^+y^+$ are valid sequences at $u$, $y^+x^+$ must also
be a valid sequence at $u$. This contradicts the assumption that $v$ was the
vertex with smallest Hamming weight on the path where $y^+$ was a valid flip.
\qed
\end{proof}

\begin{lemma}
\label{lemma-one-more-tree-like}
For $R$ a NAND-free and dual-Horn-free relation, if
$\F_1\cdot x^+\cdot y^+$ and $\F_2\cdot y^+$ are both valid positive flip
sequences at $s \in R$ such that $x^+\not\in\C(\F_2)$
then $\F_1\cdot y^+\cdot x^+$ is also valid at $s$.
\end{lemma}

\begin{proof}
Let $u = \F_1(s)$ and $v = \F_2(s)$. We apply Lemma~\ref{lemma-canonical} to the sequence
$\F_1^{-1}\cdot\F_2$ that transforms $u$ to $v$ to obtain the canonical sequence $\F = \F^+\cdot\F^-$.
Let $w$ be the vertex with maximum Hamming weight
on this canonical path (Figure~\ref{figure-lemmas-seven-eight}(b)). Hence,
we have $w = \F^+(u)$ and $v = \F^-(w)$. Note that $\F$ does not involve flips of the variables $x$ or $y$. 

Since $y^+$ is a valid flip at $v$, $y^+\not\in\C(\F^-)$, and the path from $v$ to $w$ is
monotonically increasing, from Lemma~\ref{lemma-another-nand-free}, $y^+$ is also valid at $w$.
Now using Lemma~\ref{lemma-another-tree-like}, since $x^+y^+$ is valid at $u$, $x^+\not\in\F^+$, and $y^+$ is valid at $w$, we have that $y^+x^+$ is also valid at $u$.
\qed
\end{proof}

\begin{figure}
\subfloat[]{\includegraphics{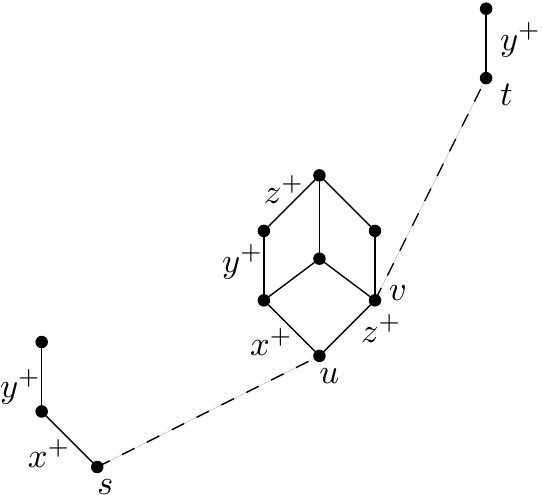}}
\hspace{3 mm}
\subfloat[]{\includegraphics[scale=1]{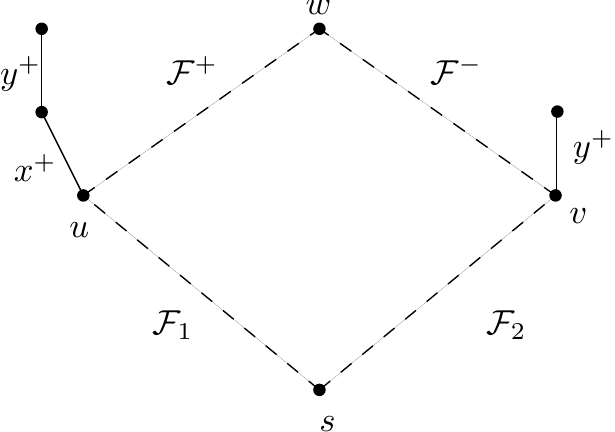}}
\caption{Dotted lines denote paths and solid lines denote edges. Hamming weight increases in the upward direction.
(a) Proof of Lemma~\ref{lemma-another-tree-like} (b) Proof of Lemma~\ref{lemma-one-more-tree-like}.}
\label{figure-lemmas-seven-eight}
\end{figure}

Lemma~\ref{lemma-union} already shows that the set of valid flip sets is closed under union. To prove that the set of valid flip sets forms a distributive lattice, we need to show that it is also closed under intersection, which we do in the next lemma.

\begin{lemma}
\label{lemma:intersection}
For $R$ a NAND-free and dual-Horn-free relation, if $\C_1$ and $\C_2$ are two positive flip
sets that are valid at $s \in R$, then $\C_1\cap\C_2$ is also a valid flip set at $s$.
\end{lemma}

\begin{proof}
If $\C_1 \subseteq \C_2$ or $\C_2\subseteq\C_1$, then the statement is trivial. Otherwise, consider any valid ordering $\F_1$ of $\C_1$. We show that if $x^+$ and $y^+$ are two consecutive elements of $\F_1$ such that $x^+\in\C_1\backslash\C_2$, $y^+\in\C_1\cap\C_2$ and $x^+<_{\F_1} y^+$, then swapping $x^+$ and $y^+$ also gives a valid ordering of $\C_1$. Applying such swaps repeatedly, we get an ordering where all elements of $\C_1\cap\C_2$ appear before all elements of $\C_1\backslash\C_2$ thus proving that $\C_1\cap\C_2$ is a valid set at $s$.

To see how to swap $x^+$ and $y^+$ in $\F_1$, suppose $u$ is the vertex on the path corresponding to $\F_1$ on which the sequence $x^+\cdot y^+$ is performed, and consider an arbitrary valid ordering $\F_2$ of $\C_2$. Let $v$ be the vertex on the path corresponding to $\F_2$ on which $y^+$ is performed. Such a vertex exists since $y^+\in\C_1\cap\C_2$. Now, since $x^+\cdot y^+$ is valid at $u$, $y^+$ is valid at $v$ and the monotonically increasing path from $s$ to $v$ does not contain the flip $x^+$ (since $x^+\in\C_1\backslash\C_2$), applying Lemma~\ref{lemma-one-more-tree-like}, we can swap $y^+$ and $x^+$ in $\F_1$.\qed
\end{proof}


The above lemma, combined with Lemma~\ref{lemma-union}, shows that the set of
valid flip sets starting at $s$ forms a distributive lattice~\cite{birkhoff1937}.
Using Birkhoff's representation theorem~\cite{birkhoff1937} on it
directly implies the next lemma. However, for clarity, we also
provide an independent proof. Let $\prec$ be a partial order defined on a set $\C$ of flips.
We say a set $\C' \subseteq \C$ is {\em downward closed} if for every $x,y \in \C$,
$y \in \C' \wedge x \prec y \implies x \in \C'$.
We say that an ordering $\F$ of a subset of elements in $\C$ \emph{obeys} the partial
order $\prec$ if (i) $\C(\F)$ is downward closed and (ii) for every $x,y \in \F$,
$x \prec y \implies x <_\F y$.

\begin{lemma}
\label{lemma-partial-order}
Let $R$ be a NAND-free and dual-Horn-free relation and $s$ be an element of $R$.
Let $\mathscr{P} = \{x^+~|~x^+\in\C$ for a positive valid flip set $\C$ at $s\}$. Then there exists a
partial order $\prec$ on $\Pc$ such that any positive flip sequence $\F$ consisting
of a subset of $\mathscr{P}$ is a valid flip sequence at
$s$ if and only if it obeys the partial order $\prec$.
\end{lemma}

\begin{proof}
Our proof proceeds by providing an explicit partial order $\prec$ on the flips in $\Pc$.
For $x^+, y^+\in\Pc$, let $x^+\prec y^+$ if and only if all valid positive flip
sequences $\F$ starting at $s$ that contain $y^+$ also contain $x^+$ and $x^+<_\F y^+$.
This is clearly a partial order since if $x^+\prec y^+$ and $y^+\prec z^+$ then $x^+\prec z^+$.

From the definition of the partial order, it is clear that every valid flip set
must satisfy the partial order. For the other direction, consider a flip
sequence $\F^*$ that satisfies the partial order. We will show that $\F^*$ is
valid by induction on the length of the flip sequence.

For the base case, $\F^*$ is trivially valid when $|\F^*| = 0$. As the induction hypothesis, suppose
that any flip sequence of length $i-1$ that satisfies the partial order is valid.
Consider the flip sequence $\F^*=(f_1,\dots ,f_i)$ that satisfies the partial order, and
let $\F_{i-1} =(f_1, \dots, f_{i-1})$.
Let $\X$ be the set of all positive flip sequences valid at $s$ whose last
element is $f_i$. Consider the set $\C=\bigcap_{\F\in\X}\C(\F)$. Since $\F^*$
satisfies the partial order, $\C\subseteq\C(\F^*)$. To see why, suppose that
$\C$ has an element $x^+$ that is not there in $\C(\F^*)$. That would mean
that $x^+$ appears before $f_i$ in all valid sequences starting at $s$. But
then $x^+ \prec f_i$ and the sequence $\F^*$ does not obey the partial order.
Thus using Lemma~\ref{lemma:intersection}, we know that $\C$ is a valid flip
set. Since $\C(\F_{i-1})$ is also a valid flip set (from the induction hypothesis), from
Lemma~\ref{lemma-union} we know that
$\C\cup\C(\F_{i-1}) = \C(\F_{i-1})\cup\{f_i\}= \C(\F^*)$ (since $\C\subseteq\C(\F^*)$) is
a valid flip set. Since $\C(\F_{i-1})$ and $\C(\F^*)$ are both valid
flip sets and $\C(\F^*) \setminus \C(\F_{i-1}) = f_i$, $\F^*$ must be a valid flip sequence.
\qed
\end{proof}

\ignore{
Suppose the set $S= \{z^+x^+y^+$, $z^+y^+x^+$, $x^+z^+y^+$, $x^+y^+z^+\}$
consisted of all valid permutations of the set $\C=\{x^+, y^+, z^+\}$
starting from $s \in R$, for some relation $R$. Clearly $z^+$ is incomparable to
$x^+$ and $y^+$; less clear is the ordering of $x^+$ and $y^+$ as it
seems to depend on the position of $z^+$. We cannot have $x^+\prec y^+$
or $y^+\prec x^+$ since both $z^+x^+y^+$ and $z^+y^+x^+$ are valid
permutations. However, if we assume that $x^+$ and $y^+$ are
incomparable, then $x^+y^+z^+$ should also have been a valid
permutation. This means there does not exist a partial order that
characterizes $S$. By the definition of $S$, $x^+y^+$ and $z^+y^+$ are
both valid permutations and $y^+x^+$ is not, which is
exactly the kind of situation that is
avoided by relations that are NAND-free and dual-Horn-free.
}

\subsubsection{Efficiently computing the shortest reconfiguration path}
We are now ready to provide a polynomial-time algorithm for
finding shortest reconfiguration paths in CNF($\mathcal{S}$) formulas where $\mathcal{S}$ is navigable.
If every relation in $\mathcal{S}$ is component-wise bijunctive, we use Gopalan et al.'s algorithm.
Otherwise, as discussed before, we assume that every relation in $\mathcal{S}$ is NAND-free and dual-Horn-free.

Let $\phi$ be a CNF($\mathcal{S}$) formula where every relation in
$\mathcal{S}$ is NAND-free and dual-Horn-free, $\{x_1,\ldots , x_n\}$
be the set of variables, and $\{C_1,\ldots , C_m\}$ be the set of
clauses in $\phi$. We wish to compute the shortest reconfiguration path between $s$ and $t$ in
$G_\phi$ for $s,t\in R_\phi$.
Let $\mathscr{P}_s$ and $\mathscr{P}_t$ be the sets of positive flips that
occur in any positive flip set valid at $s$ and $t$, respectively.

The following lemma shows that the property of any valid flip
sequence for a NAND-free and dual-Horn-free relation being
describable by a partial order, as proved in
Lemma~\ref{lemma-partial-order}, also applies to CNF($\mathcal{S}$) formulas
where every relation in $\mathcal{S}$ is NAND-free and dual-Horn-free.

\begin{lemma}
\label{lemma-formula-po}
Let $\phi$ be a CNF($\mathcal{S}$) formula where every relation
in $\mathcal{S}$ is NAND-free and dual-Horn-free.
For any $s,t\in R_\phi$, there exists a partial order $\prec_s$ on
$\mathscr{P}_s$ and a partial order $\prec_t$ on
$\mathscr{P}_t$ such that any positive flip sequence $\F_s$ consisting of a
subset of $\mathscr{P}_s$ is a valid flip sequence at $s$ if and only if
it obeys the partial order $\prec_s$
and any positive flip sequence $\F_t$ consisting of a
subset of $\mathscr{P}_t$ is a valid flip sequence at $t$ if and only if
it obeys the partial order $\prec_t$.
Moreover, $\mathscr{P}_s$, $\prec_s$, $\mathscr{P}_t$, and $\prec_t$ can be computed in polynomial time.
\end{lemma}
\begin{proof}
We compute $\mathscr{P}_s$, $\prec_s$, $\mathscr{P}_t$, and $\prec_t$ using two
directed graphs $G_s$ and $G_t$ which we construct.

We define $\mathscr{P} = \{x^+~|~x^+\in\C$ for a positive valid flip set $\C$ at $s$ for some relation in $\mathcal{S}\}$
and let $G_s$ contain a node for each flip in $\mathscr{P}$. The assignment $s$
induces an assignment $f_{X_j}(s)$ on clause
$C_j = (R_j, X_j)$ and Lemma~\ref{lemma-partial-order} defines a partial
order $\prec^j_s$ that characterizes the valid positive
sequences in $R_j$ starting at $f_{X_j}(s)$. For all $p, q\in \{1,\ldots,k_j\}$ such
that $p^+\prec^j_s q^+$, if $X_j(p)\not\in\{c_0, c_1\}, X_j(q)\not\in\{c_0, c_1\}$ and $X_j(p)\neq X_j(q)$, we
add the directed edge $(x_{X_j(p)}^+, x_{X_j(q)}^+)$ to $G_s$. We do this for
each clause $C_j$ for $j\in \{1,\ldots,m\}$. This gives us $G_s$.
Let $G_t$ be a directed graph defined similarly for $t$.

Now, in these graphs, a flip corresponding to a vertex $f$ which lies on a cycle and the flip corresponding to any  
vertex reachable from $f$ by an outgoing directed path (starting from $f$) is never going
to be performed (as the flip does not satisfy the order relation on the edges). 
Hence we remove these vertices from $G_s$ and $G_t$ as follows.
First, any vertex that appears on a directed cycle is marked to be removed. Then, we iteratively
mark every vertex that has an incoming edge from a marked vertex. Once the set of marked vertices
stops changing, we remove all marked vertices. Note that $G_s$ and $G_t$ are now acyclic.

We claim that $\Pc_s = V(G_s), \Pc_t = V(G_t)$, the partial order $\prec_s$ is such that
$f_1\prec_s f_2$ if and only if there is a directed path from $f_1$ to $f_2$ in $G_s$ and
the partial order $\prec_t$ is such that $f_1\prec_t f_2$ if and only if there is a
directed path from $f_1$ to $f_2$ in $G_t$. It is clear from Lemma~\ref{lemma-partial-order}
that any vertex that was removed in the second phase cannot be a part of any valid flip
sequence at $s$. To see that $\prec_s$ is the required partial order, it is enough to
see that any flip sequence is valid for $\phi$ if and only if it is valid for each clause.

Computing the partial orders defined by Lemma~\ref{lemma-partial-order} can be accomplished
in constant time for each relation in $\mathcal{S}$.
Then, the construction and deletion phases for $G_s$ and $G_t$
can be accomplished in polynomial time as described above.
\qed
\end{proof}

For a set $\Pc$, a partial order $\prec$ on $\Pc$, and a subset $A\subseteq\Pc$, the
\emph{smallest lower set} of $A$ is
the smallest superset of $A$ that is downward closed.
Such a lower set can be constructed in polynomial time by starting
with $A$ and including any element $f'$ not in $A$ such that $f' \prec f$ for some $f\in A$.
It is clear that any valid flip set that contains $A$ must
also contain the smallest lower set of $A$.

Now the algorithm for finding the shortest reconfiguration path is clear. We start from
$s$ and let $S$ be the set of positive flips on the variables that are
set to $1$ in $t$ and to $0$ in $s$. Then we compute the smallest lower set $S'$
containing $S$ and perform the flips in $S'$ as prescribed by the
partial order $\prec_s$ (on $\Pc_s$) to reach $s' \in R_\phi$.
We perform a similar set of flips starting from $t$ to reach $t' \in R_\phi$.
If $s' = t'$, we are done. Otherwise, we recursively find the
shortest path between $s'$ and $t'$. The complete algorithm is described in Algorithm~\ref{alg-main}.

\begin{algorithm}
\caption{\textsc{ShortestPath}($s$,$t$)}
\label{alg-main}
\begin{algorithmic}[1]
\REQUIRE A CNF($\mc{S}$) formula $\phi$ where all relations in $\mc{S}$
are NAND-free and dual-Horn-free; two satisfying assignments $s$ and $t$.
\ENSURE  Shortest reconfiguration path between $s$ and $t$.
\IF{($s = t$)}
    \STATE \RETURNR\ $\F_0$ \COMMENT{the empty flip sequence}
\ENDIF

\STATE Let $S$ be the set of positive flips that flip variables assigned 0 in $s$ and 1 in $t$.

\STATE Let $T$ be the set of positive flips that flip variables assigned 0 in $t$ and 1 in $s$.

\IF {$S$ contains an element not in $\Pc_s$ or if $T$ contains an element not in $\Pc_t$}
    \STATE \RETURNR\ Not connected.
\ENDIF

\STATE Compute the smallest lower set $S'$ of $S$ in $\Pc_s$ with respect to $\prec_s$.

\STATE Compute the smallest lower set $T'$ of $T$ in $\Pc_t$ with respect to $\prec_t$.

\STATE Let $\F_{s}$ and $\F_{t}$ be orderings of $S'$ and $T'$ that obey $\prec_s$ and $\prec_t$, respectively.

\STATE Let $s' = \F_{s}(s)$ and $t' = \F_{t}(t)$.

\STATE Let $\F = $ \textsc{ShortestPath}($s'$,$t'$).

\STATE \RETURNR\ $\F_s\cdot\F\cdot\F_t^{-1}$.
\end{algorithmic}
\end{algorithm}

We are now ready to prove the following theorem.
\begin{theorem}\label{thm:polytimealgo}
Let $\mc{S}$ be a navigable set of relations, $\phi$ be a CNF($\mc{S}$) formula, and $s$ and $t$ two of its satisfying assignments. We can compute the shortest reconfiguration path between $s$ and $t$ in polynomial time.
\end{theorem}

\begin{proof}
We show that Algorithm~\ref{alg-main} finds the shortest path between $s$ and $t$, and runs in polynomial time.
For any Boolean vector $x$, let $\eta(x)$ denote the
number of 0's in $x$ and let $\eta = \eta(s)+\eta(t)$. It is clear that
Steps 1 to 10 take time polynomial in the input size $N$, where $N = |\phi| + |\mc{S}| + |s| + |t|$. Here $|x|$ denotes the number of bits needed to represent $x$.
Since $\F_s$ and $\F_t$ are both positive flip sequences, $\eta(s')+\eta(t')\leq\eta(s)+\eta(t)-2$.
Thus the running time $T(\eta)$ of the algorithm satisfies the recursive inequality $T(\eta) \leq T(\eta - 2) + P(N)$
where $P(N)$ is some polynomial in $N$. Since $\eta < N$ the recursion solves to a polynomial in $N$.

Finally, we prove the correctness of the algorithm.
We use induction on $\eta$. If $\eta = 0$, then $s=t$ and the algorithm is trivially correct.

If the algorithm returns ``Not connected'', then it is either because of
Step 6 or Step 11. If it is because of Step 11, then by the induction hypothesis
$s'$ and $t'$ are not connected, and thus $s$ and $t$ are also not connected.
Any flip sequence that transforms $s$ to $t$ must perform each flip in $S$. Thus it is also clear
that if Step 6 returns ``Not connected'', then $s$ and $t$ are not connected.

If the algorithm returns a flip sequence, then we claim that it is a shortest sequence.
From induction, we know that $\F$ is a shortest flip sequence from $s'$ to $t'$.
The claim follows from the observation that if $s$ and $t$ are connected, then there
must exist a shortest path from $s$ to $t$ that passes through both $s'$ and $t'$.
Let $\F_1\cdot\F_2^{-1}$ be a shortest flip sequence from $s$ to $t$ such that $\F_1$
and $\F_2$ are both positive. It is clear that $S'\subseteq\C(\F_1)$. Since $S'$ itself
is valid, from Lemma~\ref{lemma-formula-po}, there must exist a valid ordering
of $\C(\F_1)$ that first performs all flips of $S'$. In this ordering, the vertex
reached after performing all flips of $S'$ is exactly $s'$. Using a similar argument
on $\F_2$, we get a shortest path that goes through both $s'$ and $t'$.
\qed
\end{proof}

\section{Final remarks}
Many problems can be modelled as finding shortest paths in large graphs. Our result provides new insights into the kinds of structures a graph will need to possess to be amenable to an efficient shortest path algorithm. The fact that the shortest path in navigable formulas flips variables that are not in the symmetric difference is evidence that our algorithm exploits a property of the reconfiguration graph that is fundamentally new. Any previously known properties that were used to find shortest paths efficiently also rendered the graph too simple, in that any shortest path only flipped the symmetric difference. It will be interesting to see if our results help us understand other large graphs, in particular, the flip graph of triangulations of a convex polygon where the complexity of finding the shortest path is still open.

\bibliographystyle{abbrv}
\bibliography{satrefs}
\end{document}